\documentclass[10pt, journal]{IEEEtran}

\usepackage{multirow}
\usepackage{graphicx,comment,amssymb,amsthm} 
\usepackage{tabularx}
\usepackage{url}
\usepackage{color}
\usepackage{epstopdf}
\usepackage{autonum}
\usepackage{enumitem}
\usepackage{dsfont}

\newcommand{\fref}[1]{Fig.~\ref{#1}}
\newcommand{\eref}[1]{(\ref{#1})}
\newcommand{\sref}[1]{Sec.\,\ref{#1}}

\newcommand{\thref}[1]{Theorem\,\ref{#1}}

\newtheorem{theorem}{Theorem}

\newtheorem{proposition}{Proposition}

\setlist[itemize]{leftmargin=*}
\setlist[enumerate]{leftmargin=*}
\newcommand{\PP}{\mathds{P}} 
\newcommand{\EE}{\mathds{E}} 

\begin{document}

\title{Characterizing Delay and Control Traffic of the Cellular MME with IoT Support
\thanks{A sketch of the present work was included in our poster presented at ACM Mobihoc 2019 \cite{nostro-poster}.}
}

\author{Christian~Vitale,~\IEEEmembership{Member,~IEEE,}
        Carla Fabiana~Chiasserini,~\IEEEmembership{Fellow,~IEEE,}
        Francesco~Malandrino,~\IEEEmembership{Senior~Member,~IEEE,} and~Senay~Semu~Tadesse,~\IEEEmembership{Member,~IEEE}

\IEEEcompsocitemizethanks{\IEEEcompsocthanksitem C.~Vitale is with KIOS Center of Excellence, Cyprus. C.~F.~Chiasserini and S. Tadesse are  with Politecnico di Torino, Italy. C. F.~Chiasserini and F. Malandrino are with CNR-IEIIT, Italy. C.~F.~Chiasserini and F.~Malandrino are also with CNIT, Italy.
} 
}

\maketitle

\begin{abstract}
One of the main use cases for advanced cellular networks is represented by  massive Internet-of-things (MIoT), i.e., an enormous number of IoT devices that transmit data toward the cellular network infrastructure. To make cellular MIoT a reality,  data transfer and control procedures specifically designed for the support of IoT are needed.   
For this reason,   3GPP has introduced the Control Plane Cellular IoT optimization, which foresees a simplified bearer instantiation, with  the Mobility Management Entity (MME)
handling both control and data traffic. The 
performance of the MME has therefore become  critical, and properly scaling its 
computational capability can determine the ability of the whole network to tackle
 MIoT effectively. 
In particular, considering virtualized networks and the need for an efficient allocation of computing resources, it is paramount to  characterize the MME performance as the MIoT traffic load changes.
We address this need by presenting compact,
closed-form expressions linking the number of IoT sources with the rate at
which bearers are requested, and such a rate with the
delay incurred by the IoT data. We show that our analysis, supported by testbed experiments and verified through large-scale simulations, represents a valuable tool to make effective  scaling decisions in virtualized cellular core networks.
\end{abstract}

\section{Introduction} 

Massive Internet-of-things (MIoT) is an umbrella term for a fairly diverse set of applications, including smart factory, cloud robotics, automotive leveraging smart city sensors, and surveillance/security; as such, it represents one of the main motivations behind 5G~\cite{5gppp-usecases}. For all these applications, service latency is a critical constraint, even more so than sheer network throughput. Also, IoT applications are characterized by  a very  high density of devices, up to 10,000~devices/km$^2$~\cite[p.~6]{5gppp-usecases}, and peculiar traffic patterns: devices may be inactive for a long time, and then multiple devices may transmit data in a (almost) synchronized manner.

Such traffic patterns are a poor match for the default procedures followed by the cellular core network  and such a mismatch may jeopardize the application-latency requirements. Indeed, before a terminal can transmit data packets toward the cellular infrastructure, typically the following operations are required:  authentication,  identity verification, and bearer establishment.  If the terminal remains silent longer than a timeout, the bearer is released and the whole procedure has to be performed again. Thus, for  MIoT traffic, bearer instantiation (including bearer establishment and release) is one of the most critical tasks: using the default procedures would result in an exceedingly high latency and  control overhead, compared to the data traffic  generated by an MIoT device. 

To cope with that, 3GPP has introduced a new standard~\cite{3gppepc}, called Control Plane Cellular IoT Evolved Packet System (CIoT) optimization, which is  already available in off-the-shelf products \cite{cisco}. Such a standard  (i) simplifies the procedures, roughly halving the associated overhead, (ii) uses the  Mobility Management Entity (MME) of the cellular core network to forward user-plane traffic, 
and (iii) limits the involvement of MIoT sources in bearer establishment procedures, hence reducing the  power consumption. Importantly, since under the CIoT optimization the MME is in charge of both control- and user-plane processing,  it bears the brunt of MIoT traffic, thus becoming  the pivotal component of the cellular core. It follows that the MME performance and the associated delay determine the ability of the network as a whole to support MIoT traffic.

Ensuring that the MME has sufficient computational capability to efficiently process the traffic load generated by MIoT sources becomes even more sensitive in the context of network {\em softwarization}. Such a paradigm refers to a global trend towards replacing special-purpose network equipment -- including the entities of the cellular core~\cite{baumgartner2015mobile} -- with virtualized network functions (VNFs) running on general-purpose hardware. In the case of a virtual Evolved Packet Core (EPC) \cite{v-epc}, the number of MME instances and their computational capability can be {\em scaled} to adapt to the variations in the current and expected MIoT traffic they must process. In particular, in the case of the MME, effective scaling requires:
\begin{itemize}
\item characterizing the relation between the number of MIoT sources and the arrival rate
of bearer requests at the MME;
\item modeling the impact of the MME capacity on the delay introduced by the bearer establishment procedure.
\end{itemize}

In this paper, we study both the above aspects with reference to the case where a network operating according to the CIoT optimization serves MIoT traffic. Specifically, our main contributions are  as follows:

\begin{itemize}[leftmargin=0.1cm]
    \item[{\em (i)}] We begin by characterizing analytically the time between consecutive bearer instantiation requests coming from MIoT sources, proving that it is well described by an exponential distribution;
    \item[{\em (ii)}] By running and profiling the components of a real-world EPC implementation, we make some fundamental observations on the system that we then exploit to develop our analytical model;
	\item[{\em (iii)}] Leveraging the analytical results on MIoT traffic patterns and the experimental observations,  
we build an M/D/1-PS queuing model of the MME and study how the bearer instantiation time depends on (a) the traffic load, i.e., the arrival rate of bearer requests, and (b) the computational capability assigned to the MME itself. Importantly, we obtain a novel {\em closed-form} expression for the packet forwarding delay as a result of our analysis;
	\item[{\em (iv)}] We show that the obtained analytical results represent a powerful tool to drive real-time scaling decisions in softwarized cellular networks when dealing with delay-sensitive applications;
	\item[{\em (v)}] We validate our analysis through large-scale simulations using both a synthetic traffic model based on the 3GPP standard, and a real-world scenario including topology and mobility information from the city of Monte Carlo, Monaco.
\end{itemize}

The remainder of the paper is organized as follows. After introducing our system model in Sec.~\ref{sec:model}, we present our analysis and a closed-form expression for the characterization of the bearer request arrival process in Sec.~\ref{sec:input_MME}.  In Sec.\,\ref{sec:epc-model}, we run some experiments and make  useful observations to develop our analytical model of the cellular core network. Furthermore,   we  characterize the core network delay performance. 
Through   detailed simulations using both synthetic and real-world traffic, in Sec.~\ref{sec:results} we show how our analysis can be used to effectively tune the computational capability of a  vEPC.   
Finally, we review related work in Sec.~\ref{sec:related} and conclude the paper in Sec.~\ref{sec:conclusions}.

\section{System Model and Preliminaries}\label{sec:model}

Here, we present the CIoT bearer instantiation procedure~\cite{3gppepc} and how the MIoT traffic is served when such a procedure is adopted. In particular, we consider the current cellular network, namely, the Evolved Packet Core (EPC), which is   briefly introduced  in Sec.\,\ref{subsec:enhancements_IOT}. Then we detail the CIoT bearer instantiation procedure and its relevance to the NB-IoT standard in Sec.\,\ref{subsec:bearer_diagram}. Finally, we describe the model we adopt for the IoT traffic, in  Sec.\,\ref{subsec:iot_traffic_model}.  

\subsection{Evolved packet core network}
\label{subsec:enhancements_IOT}

IoT cellular traffic has to traverse the EPC network, which includes four main components, as depicted in \fref{EPC}:
\begin{itemize}
\item the Serving Gateway (S-GW) mainly routes data traffic and acts as anchor point when User Equipments (UEs) move from one eNB to another;
\item the PDN Gateway (P-GW) acts as ingress and egress point of the mobile access network; 
it  is also the responsible for policy enforcement; 
\item the Mobility Management Entity (MME) is the termination point of UE control channels. The MME authenticates and tracks registered UEs and, most importantly, it  handles bearer activation, i.e., it is the MME that creates a data path between the UEs and the P-GW. When CIoT optimization is in place, the data path between the UE and the P-GW includes the MME itself, since the MME is also responsible for relaying the traffic of the MIoT sources to the correct S-GW (see \fref{EPC}); 
\item the Home Subscriber Server (HSS) is a central database where UE-related information is stored. The HSS assists the MME in UE authentication.
\end{itemize}
Note that the MME is connected  to the S-GWs for bearer establishment and, under the CIoT optimization, it also performs packet decryption/forwarding, while the P-GW handles the data traffic to/from several S-GWs. Importantly, in the case of a vEPC, the MME, P-GW, and S-GW typically run on different (virtual) machines whose number and capability can be adjusted as needed.

\begin{figure}
\centering
	\includegraphics[width=1\columnwidth]{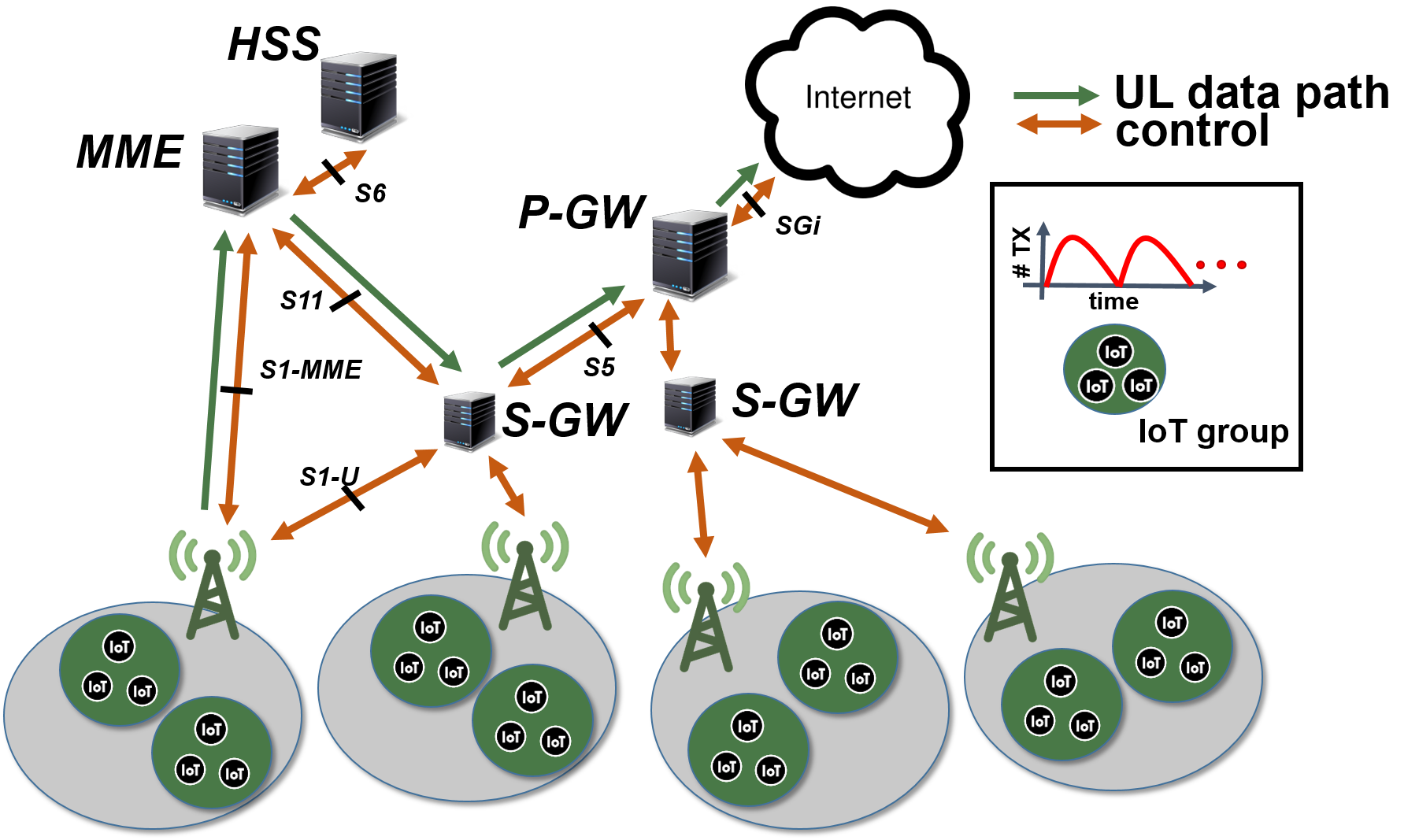}
\caption{EPC architecture.}
\label{EPC}
\end{figure}

\subsection{CIoT bearer instantiation procedure}
\label{subsec:bearer_diagram}
Exactly as any other cellular transmitter, an MIoT source sends or receives data traffic only if a logical connection with the corresponding P-GW is in place, i.e.,  if the MME has completed the bearer instantiation procedure. However, unlike the ordinary procedure, the CIoT optimization foresees that bearers are released immediately after packet transmission/reception, unless an MIoT source  explicitly signals the presence of imminent traffic. As a consequence, an established  bearer lasts for quite a short time and no handover procedure is typically required  for MIoT traffic at the MME level. In the following, we therefore focus only on the performance of the MME when handling bearer instantiation procedures.

As depicted in \fref{diagram},  
each time an MIoT source has to transmit a packet, five operations are performed: (i) authentication, (ii) identity verification, (iii) bearer establishment, (iv) forwarding (after data decryption and integrity check) of the data packets piggybacked by the MIoT source in the Radio Resource Control (RRC) Early Data Request message, and (v) bearer release. 
Specifically, hereinafter   {\em bearer establishment}  will  refer to the set of operations comprised between step 1 and step 6 (included) in \fref{diagram}. 
 We remark that such a procedure represents a crucial contribution to the data forwarding  latency, and it cannot be overlooked in the MIoT data delay computation.  Indeed, the time needed to complete a bearer establishment also corresponds to the delay incurred by the data  transfer within the EPC.

\begin{figure}
\centering
	\includegraphics[width=1.03\columnwidth]{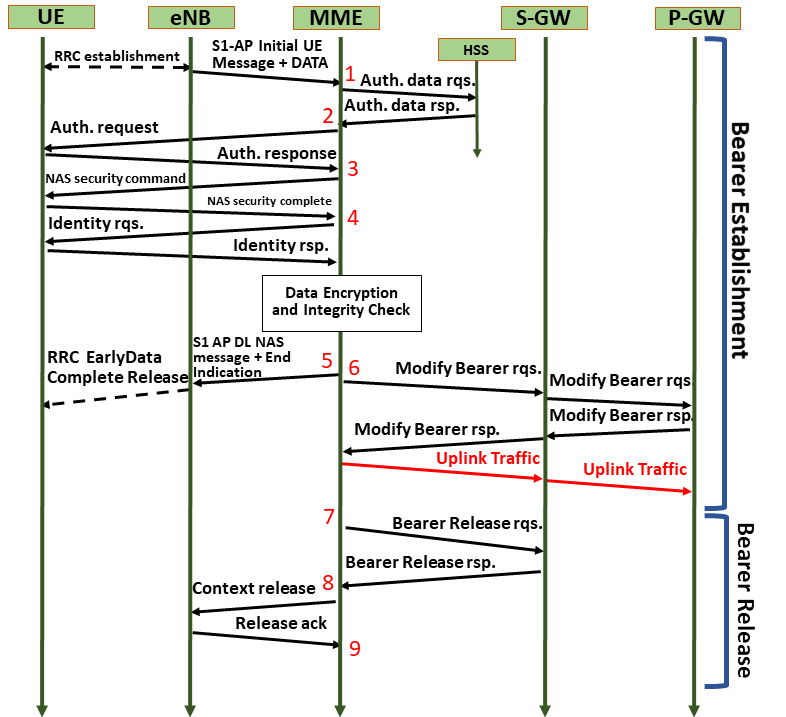}
\caption{CIoT bearer instantiation procedure and uplink data transmission.}
\label{diagram}
\end{figure}

Finally, it is worth remarking that CIoT well pairs up with the Narrowband IoT (NB-IoT) standard -- both being specifically designed to support massive IoT traffic, taking, respectively, the core network side and the radio access perspective. Indeed, NB-IoT is an IoT system built from existing LTE functionalities, which aims to support over 50,000 low-data rate stationary devices within a cell-site sector \cite{NB-IoT-1}. 
 Beside defining an energy-efficient and robust physical layer for enhanced indoor coverage, NB-IoT also effectively addresses cell search, synchronization, and random access for initial link establishment. Specifically,  according to NB-IoT, the RRC establishment on the top-right of \fref{diagram}  summarizes the following steps \cite{NB-IoT-2}:
(a) the UE transmits a random access preamble; (b) the eNB replies with a random access response including a timing advance command and which uplink resources are assigned the UE to perform (c); (c) the UE transmits its identity;  (d) the eNB transmits a message to resolve any contention due to multiple UEs accessing the channel (in step (a))  using the same preamble.

\subsection{IoT traffic model}
\label{subsec:iot_traffic_model}

As mentioned above, after data transmission/reception, the MIoT source's bearer is released and a new bearer has to be established if later on the MIoT source has some more traffic to send/receive. Intuitively, depending on the IoT traffic pattern, the time between subsequent data packets may vary significantly, and so does the rate of bearer instantiation requests of an MIoT source. 
In order to characterize the arrival process of bearer requests, in the following we consider the  traffic model described by the 3GPP standard \cite{3gppmtc} for machine type communications.

In~\cite{3gppmtc}, MIoT sources are organized in groups. Reflecting real-world operation conditions, \cite{3gppmtc} envisions quasi-synchronous packet transmissions within a group. This represents, for example, a group of sensors monitoring a geographical area, programmed to raise an alarm when a specific event occurs. After the occurrence of the event of interest, e.g., a gas leak, the closest sensors to the event raise the alarm. Sensors neighbouring the area where the event occurred react to the event subsequently, with a delay due to the propagation of the phenomenon. Such an effect triggers alarms from all the sensors belonging to the group, with a peak of alarms (hence, of traffic) roughly at the center of a period and an aggregate traffic distribution over time that follows a Beta(3,4). In~\cite{3gppmtc}, the events, 
and the related group transmissions,  occur in subsequent periods of duration $T$, each of which with an aggregate traffic distribution over time following a Beta(3,4).

Notice that this model is quite general. Indeed, modifying $T$ allows us to account for different aggregate transmission rates, while setting the group size to 1 allows us to represent MIoT sources behaving independently from each other. Furthermore, as we explain later in \sref{sec:input_MME}, the model can be easily adapted to include data aggregators (a.k.a. gateways) that, as often envisioned in sensor network applications, collect and forward the data packets generated within groups of MIoT devices.

The traffic model specified by the 3GPP standard represents the aggregate behavior of a set of MIoT sources. However, we are interested in characterizing the latency of the data transfers by individual sources, each of which requires a bearer instantiation.  In order to address this issue, we leverage the data generation model  for  an individual IoT device  presented in \cite{laner2013traffic}, which results in an    aggregate group traffic that still matches the Beta(3,4) distribution specified by the 3GPP standard. 

In \cite{laner2013traffic}, each MIoT source is modeled as a Markov chain including two states, named {\em regular operation} and {\em alarm}, and hereinafter denoted with $R$ and $A$, respectively. The period $T$ of the IoT traffic pattern is divided into an arbitrary number~$N$ of slots, each of duration $\delta$. 
In state $A$, the MIoT source sends packets according to a Poisson process with mean $\lambda_A$, which, without loss of generality and consistently with \cite{laner2013traffic}, we set  equal to $1$ packet/slot, i.e., an MIoT source successfully transmits at least one data packet with probability $(1-e^{-1})$.  Introducing such a probability of transmitting (at least) one packet while being in state $A$ allows capturing communication aspects that may arise in real-world  IoT scenarios, e.g., transceiver failure or harsh propagation conditions due to the unfavourable IoT location. In state $R$, instead, the MIoT source transmits packets with an arbitrary small rate $\epsilon$, representing, e.g., keep-alive or synchronization messages (in \cite{laner2013traffic} the average transmission rate in state $R$ is set  to $\lambda_R \mathord{=} 0.0005$ packet/s).

In each slot $n$ ($n= 1,...,N$) within a period $T$, the MIoT source may move from one state to the other. When in $A$, the source moves to $R$ in the next time slot with probability $1$. When in $R$, the source moves to $A$  in time slot $n$ with probability mass function (pmf)\footnote{The pmf of a discrete random variable $x$ at $n$ will be denoted by $f_x(n)$. The evaluation at $n$ of the pmf of $x$ conditioned to the random variable $y$, when $y=m$, will be denoted by $f_x(n|y\mathord{=}m)$
Also, we will indicate with $\PP(X)$ the probability of a specific event $X$. The probability density function (pdf) and the cumulative density function (CDF) of a continuous random variable $x$ will be denoted by $f_x(y)$ and $F_x(y)$, respectively.} $f_b(n)$, which depends on the considered slot in the period.  As shown in  \cite{laner2013traffic}, $f_b(n)$ can be obtained from the sampling of the Beta(3,4) shape\footnote{Note that the Beta(3,4) distribution is only defined in $[0,T]$.}, as follows: 
\begin{equation}
\label{eq_probability}
f_b(n) = \text{Beta}\left(\frac{n \delta}{T}\right)\frac{\delta}{T} = 60\left(\frac{n \delta}{T}\right)^2 \left(1\mathord{-}\frac{n \delta}{T}\right)^3\frac{\delta}{T} \,.
\end{equation}

In summary, given the above IoT traffic model, on average, an MIoT source visits state $A$ once every $T$ seconds, and therein it transmits a packet with probability  $(1-e^{-1})$. Instead, when an MIoT source sojourns in state $R$, it transmits a packet every $1/\epsilon$ s. In the following, we assume that $\epsilon$ is small enough so that we can neglect the occurrence of a packet transmission in state $R$  (this holds, e.g., setting $\epsilon$ to the value suggested in \cite{laner2013traffic}).

Finally, from \eref{eq_probability}, we observe that an MIoT source moves from state $R$ to state $A$ with a probability that only depends on slot $n$, not on the past, nor on the activity of other MIoT sources. It is important to point out two aspects that justify the use of such a model: \\
\noindent
(i) even if in the model the single IoT traffic does not depend on the past, the obtained overall aggregate traffic of a group still follows the Beta(3,4) distribution suggested by the 3GPP specification, i.e., the group aggregate traffic follows the typical
pattern of a set of sensors reacting to a specific event; \\
\noindent
(ii) from the perspective of the core network, it does not matter which IoT sensor within a group triggers an alarm; indeed, each IoT transmitter within a group performs the same bearer instantiation procedure and introducing spatial correlation between IoT sensors activity does not have any impact on the MME load and the resulting distribution of the packet forwarding latency.


\begin{table}[h!]
\caption{Table of notations \label{tab:notations}}
\centering
\begin{tabular}{|| c | c ||} 
 \hline
 Symbol & Variable \\
 \hline
  \hline
 $\delta$   & slot duration \\
  \hline
 $T$        & time between events monitored by MIoT groups \\  
 \hline
 $N$        & number of slots in a period $T$ \\
 \hline 
 $Q$        & number of MIoT sources served by the EPC \\
 \hline
 $f_b(n)$ ($f_b(n)$)   & probability  of transition from $R$ \\
&             to $A$ in slot $n$ (time $t$) \\ 
 \hline
 $\beta$    & time between bearer requests at the EPC \\
 \hline 
 $s$        & slot (time) of the last bearer request \\ 
 \hline 
 $s_q$      & slot (time) of the last bearer request by  source $q$ \\ 
 \hline 
 $\alpha$   & time between the last bearer request and \\ 
            & the next transition to $A$ by any source \\
 \hline
 $\alpha_q$ & time between the last bearer request and  \\ 
            & the next transition to $A$ by source $q$ \\             
 \hline
 $\omega_q$ & offset of the time reference of source $q$ \\ 
            & with respect to source $0$ \\  
 \hline
 $E(z,\lambda_\alpha)$ & Erlang CDF with shape $z$ and rate $\lambda_\alpha$\\
 \hline                  
  $\lambda_x$ & rate of the exponential random variable $x$\\
 \hline 
  $O_X$  	& number of  CPU operations per bearer \\
&  procedure  for EPC entity $X$\\
 \hline 
  $C_X$  	& capacity, in CPU operations/s, of entity $X$\\
 \hline 
  $d$ 		& time between a bearer request and its completion, \\
 \hline
  $v$ 		& delay due to the MME of the \\
			& bearer establishment  procedure \\
 \hline 
  $K$ 		& constant delay due to all EPC entities, other \\
			& than the MME,  in bearer establishment  \\
 \hline 
\end{tabular}
\end{table}

\section{IoT Control Traffic Characterization}
\label{sec:input_MME}

To evaluate the delay performance of the MME when the CIoT optimization is supported,  we first prove that the  arrival process of the bearer instantiation requests at the MME follows a Poisson distribution.  
To this end, in this section we  derive $F_\beta(\cdot)$, the cumulative distribution function (CDF)   of the time interval between subsequent bearer instantiation requests at the MME. The steps we perform are summarized below:\\
\noindent
{\em (i)} we observe that, under the CIoT optimization, every time an MIoT source has a new packet to transmit, the MME has to establish a new bearer and forward the packet to the right S-GW. Thus, the time interval between subsequent bearer instantiations by the MME corresponds to the time interval between packet transmissions by any of the MIoT sources served by the MME; 
\vskip 1mm
\noindent
{\em (ii)} we then derive $F_\alpha(\tau | s=t)$, the distribution of the time interval between a packet transmission by any source in the system  and the subsequent visit to state $A$ by any, potentially different, MIoT source;
\vskip 1mm
\noindent
{\em (iii)} for $\delta\rightarrow 0$, we prove  that such a distribution does not depend on the time of the last transmission in the system and turns out to be exponential. Furthermore, the result holds also for the  time interval between subsequent packet transmissions, i.e., the  inter-arrival time of bearer requests  at the MME.  

All notations we adopt are summarized in Table\,\ref{tab:notations}; we also mention that the term  ``packet transmission'' is often used interchangeably with ``bearer request''. 

\subsection{Inter-arrival time between bearer requests}

In the following, we consider $Q$ MIoT sources served by the same MME, generating traffic according to the  3GPP model described in \sref{subsec:iot_traffic_model}. 
As the first step,  we fix to $k$ the time slot at which the last transmission in the system occurred and we compute $f_\alpha(m|s\mathord{=}k)$, i.e., the probability density function (pdf) of the time interval between $k$ and the slot in which the first device,  among the $Q$ MIoT sources,  moves to state $A$. 
It is easy to see that $f_\alpha(m|s\mathord{=}k)$ can be written as the minimum over the time intervals between $k$ and the first visit to $A$ of the $Q$ MIoT sources, i.e.,   
\begin{equation}
\label{iot_alpha_distribution}
f_\alpha(m|s\mathord{=}k) = f_{\min (\alpha_q)}(m|s\mathord{=}k) \,.
\end{equation}
In the above expression, $\alpha_q$ is the time interval between $k$ and the transition to state $A$ of the  MIoT source $q$, and $f_{\min (\alpha_q)}(\cdot)$ is the pdf of the minimum over the $\alpha_q$'s. 

Using (\ref{iot_alpha_distribution}) and considering the fact that in the adopted MIoT traffic model, MIoT packet transmissions are independent of each other, the CDF $F_\alpha(m|s\mathord{=}k)$ can be obtained as the minimum among random variables:
\begin{equation}
\label{iot_MME_distribution_CDF}
F_\alpha(m|s\mathord{=}k) = 1\mathord{-}\prod_{q\mathord{=}1}^Q \left ( 1\mathord{-}F_{\alpha_q}(m|s\mathord{=}k) \right ).
\end{equation}

As already mentioned, an MIoT source moves from state $R$ to state $A$ with a probability that depends only on the slot within period $T$  corresponding to time $k$, i.e., on  $k$ only and not on the past. Thus, in  the following proposition, we can  prove that \eref{iot_MME_distribution_CDF} can be computed as if any MIoT source $q$  transmitted its last packet in $k$,  i.e., denoting with $s_q$ the slot of the last packet transmission\footnote{Since different groups are not syncronized with each other, i.e., time $k$ corresponds to different slots within the period of different MIoT sources, $F_{\alpha_q}(m|s_q\mathord{=}k)$ depends on MIoT source $q$.} by $q$, $F_{\alpha_q}(m|s\mathord{=}k) = F_{\alpha_q}(m|s_q\mathord{=}k)$, $\forall q$. This is an important property, which  allows us to greatly 
  simplify the subsequent derivations. 

\begin{proposition}
\label{lemma_alpha}
Denote with $s$ the variable representing the slot in which the last packet transmission in the system by any of the MIoT sources occurred, and with $\alpha_q$ the time interval between slot $s$ and the subsequent transition of the $q$-th MIoT source to state $A$. Since the IoT traffic does not depend on the past, $\alpha_q$ can be computed as if the last packet transmission in the system was by $q$. Denoted with $s_q$ the slot of the last packet transmission by $q$, for a sufficiently large number of slots in period $T$, i.e. for a small $\delta/T$, we get: 
\begin{equation}
\label{second_property}
F_\alpha(m|s\mathord{=}k)  = 1\mathord{-}\prod_{q\mathord{=}1}^Q \left ( 1\mathord{-}F_{\alpha_q}(m|s_q\mathord{=}k) \right ).
\end{equation}
\end{proposition}
\begin{proof}
The proof of the proposition can be found in Appendix A in the Supplemental Material.
\end{proof}

The above proposition tells us that  $F_\alpha(m|s\mathord{=}k)$ can be derived by analyzing the dynamics of the individual MIoT sources separately, i.e., through the CDF of the time interval between the last transmission by $q$ and the subsequent visit to state $A$ of $q$ itself, which is significantly easier to compute than   using $F_{\alpha_q}(m|s\mathord{=}k)$.

Next, we rewrite $F_{\alpha_q}(m|s_q\mathord{=}k)$  accounting for the time reference of source $q$. To this end, we recall that 
each source belongs to a specific group and it is quasi-synchronized only with the IoT sources belonging to that group, while different groups 
may exhibit a temporal offset with respect  to each other \footnote{Sources belonging to the same group have zero offset relatively to each other.}. 
By taking as global reference the time of source 0,  we denote with $\omega_q \in\{0,...,N\mathord{-}1\}$ the time offset between source $0$ and the $q$-th source ($q=1,\ldots,Q-1$). Then 
\eref{second_property} can be rewritten as: 
\begin{equation}
\label{second_property_offset}
F_\alpha(m|s\mathord{=}k)   =   1\mathord{-}\prod_{q\mathord{=}1}^Q \left ( 1\mathord{-}F_{\alpha_q}(m|s_q(k,\omega_q)) \right ),
\end{equation}
where  $s_q(k,\omega_q)\mathord{=}\mathrm{mod}(k\mathord{+}\omega_q,N)$ and 
\begin{eqnarray}
\label{eq_alarm}
F_{\alpha_q}(m|s_q(k,\omega_q)) & \mathord{=} & \sum_{x=1}^m f_b(\mathrm{mod}(s_q(k,\omega_q) \mathord{+}x,N))  \cdot \nonumber \\
& & \prod_{y\mathord{=}s_q(k,\omega_q)\mathord{+}1}^{s_q(k,\omega_q)\mathord{+}x-1} \kern-1em \left[ 1 \mathord{-} f_b(\mathrm{mod}(y,N))   \right],
\end{eqnarray}
with $f_b(n)$ being the transition probability from $R$ to $A$ given in (\ref{eq_probability}). In (\ref{eq_alarm}), $F_{\alpha_q}(m|s_q(k,\omega_q))$ has been derived considering the probability that the following sequence of  events takes place: no transition into state $A$ for $m\mathord{-}1$ slots, and  a transition into state $A$, exactly $m$ slots after $s_q(k,\omega_q)$. 


We now switch to continuous time and evaluate the system dynamics when the slot duration $\delta$ tends to 0. We recall that the duration of slot $\delta$ is arbitrary and it only affects the number of slots within a period of activity of a group, without affecting the MIoT traffic model. 
Let us denote with $t$ the reference time of MIoT source $0$, with $s_q(t,\omega_q)$ the time instant of MIoT source $q$ in its period corresponding to $t$, i.e., $s_q(t,\omega_q)\mathord{=}\mathrm{mod}(t\mathord{+}\omega_q,T)$, and with $\tau$ the interval from the last transmission in the system to the time of the first transition from $R$ to $A$ by any of the MIoT sources.

First, for $\delta\rightarrow 0$, we  rewrite \eref{second_property_offset} and \eref{eq_alarm}, respectively,  as,
\begin{equation}
\label{second_property_offset_continous}
F_\alpha(\tau|s\mathord{=}t)  =   1\mathord{-}\prod_{q\mathord{=}1}^Q \left ( 1\mathord{-}F_{\alpha_q}(\tau|s_q(t,\omega_q)) \right ) 
\end{equation}
and 
\begin{eqnarray}
\label{eq_alarm_2}
F_{\alpha_q}(\tau|s_q(t,\omega_q)) &\mathord{=} &\int_0^\tau f_b\left(\mathrm{mod}(s_q(t,\omega_q)\mathord{+}x,T)\right)\mathord{\cdot} \nonumber \\
& & {\displaystyle \prod_{y\mathord{=}s_q(t,\omega_q)}^{s_q(t,\omega_q)\mathord{+}x}}\kern-1em \left(1\mathord{-}f_b(\mathrm{mod}(y,T)\right)dx ,
\end{eqnarray}
where $f_b(x)$ can be obtained directly from \eref{eq_probability}  as, 
\begin{equation}
\label{beta_continous}
f_b(x) = \frac{60\left(\frac{x}{T}\right)^2\left(1-\frac{x}{T}\right)^3}{T}
\end{equation}

Looking at \eref{second_property_offset_continous}, one can see that, when the number of MIoT sources in the system grows, the time interval between a packet transmission and the subsequent visit to state $A$ by any MIoT source decreases dramatically, since 
the minimum over a large number of positive random variables should be considered.  Consequently, it is enough to provide an expression for $F_{\alpha_q}(\tau|s_q(t,\omega_q))$ that is accurate for small values of $\tau$; given that, we can assume: $s_q(t,\omega_q)\mathord{+}\tau<T, \forall~s_q(t,\omega_q)\in[0,...,T]$. Then a good approximation of $F_{\alpha_q}(\tau|s_q(t,\omega_q))$ for $Q$ large, hence $\tau$ small, is given by:
\begin{equation}
\label{eq_alarm_alpha_continous}
F_{\alpha_q}(\tau|s_q(t,\omega_q))   = \int_0^\tau f_b(s_q(t,\omega_q)\mathord{+}x){\displaystyle \prod_{y\mathord{=}s_q(t,\omega_q)}^{s_q(t,\omega_q)\mathord{+}x}}\kern-0.5em \left(1\mathord{-}f_b(y)\right)dx \,.
\end{equation}

Interestingly, the product form in \eref{eq_alarm_alpha_continous} is the Volterra's product integral.   
Using such an integral expression in \eref{eq_alarm_alpha_continous}, we obtain: 
\begin{equation}
\label{F-V}
F_{\alpha_q}(\tau|s_q(t,\omega_q)) = \int_0^\tau f_b\left(s_q(t,\omega_q)\mathord{+}x\right) e^{{\displaystyle \int_{s_q(t,\omega_q)}^{s_q(t,\omega_q)\mathord{+}x}}\kern-2.5em\mathord{-}f_b\left(y\right) dy}dx.
\end{equation}

Replacing \eref{beta_continous} in \eref{F-V} and solving both integrals, we get:
\begin{eqnarray}
\label{eq_alarm_6}
F_{\alpha_q}(\tau|s_q(t,\omega_q)) & \mathord{=} & 1-e^{ \mathord{-}\frac{\tau}{T}  
\left( 60\left(\frac{s_q(t,\omega_q)}{T} \right)^2 \left( 1 \mathord{-} \frac{s_q(t,\omega_q)} {T} \right)^3 \right) }  \nonumber\\
&&\cdot \, e^{o(\tau^2) }  \nonumber\\
&{\stackrel{(a)}{\approx}}& 1- e^{- f_b(s_q(t,\omega_q))\tau} \,,
\end{eqnarray}
where $(a)$ holds for $\tau$ small.
As a result, for $Q$ large, $F_{\alpha_q}(\tau|s_q(t,\omega_q))$ follows an exponential distribution with rate parameter $f_b(s_q(t,\omega_q))$. Substituting \eref{eq_alarm_6} in \eref{second_property_offset_continous}, we obtain:
\begin{eqnarray}
\label{fbeta_theorem}
F_\alpha(\tau|s\mathord{=}t)  &\approx & 1 \mathord{-} \exp \left (- \sum_{q\mathord{=}1}^Q f_b(s_q(t,\omega_q)) \tau \right) \nonumber \\
&=& 1\mathord{-} e^{-\lambda_{\alpha|t} \tau} \,, 
\end{eqnarray}
which states that,  when $Q$ grows large, $F_\alpha(\tau|s\mathord{=}t)$ follows an exponential distribution  with rate $\lambda_{\alpha|t} = \sum_{q\mathord{=}1}^Q f_b(s_q(t,\omega_q))$. 

Interestingly, under the above conditions, 
we can write:
\begin{eqnarray}
\label{eq_proof}
\lambda_{\alpha|t} &\approx & Q \int_0^T \PP(s_q(t,\omega_q)) f_b(s_q(t,\omega_q))d\omega_q \nonumber \\
&= & Q \int_0^T \frac{f_b(s_q(t,\omega_q))}{T} d\omega_q \nonumber \\
&=& \frac{Q}{T} \nonumber \\
& \triangleq & \lambda_\alpha \,,
\end{eqnarray} 
where we considered that: 
\begin{itemize}
\item  exploiting the law of large numbers, the experienced $\lambda_{\alpha|t}$ is approximated accurately by its average value; 
\item  for high values of $Q$, also the number of IoT groups served by the MME grows large, hence the offsets $\omega_q$ can be assumed to be random variables uniformly distributed in $[0,T]$; 
\item such an observation holds also for $s_q(t,\omega_q)$, $\forall t$, since, by definition, $s_q(t,\omega_q)=\mathrm{mod}(t\mathord{+}\omega_q,T)$; 
\end{itemize}
Note that \eref{eq_proof}  not only  states that $F_\alpha(\tau|s\mathord{=}t)$  follows an exponential distribution, but also that such a distribution does not depend on $t$, i.e., $F_\alpha(\tau|s\mathord{=}t)\mathord{=}F_\alpha(\tau)$.

We now use this result to compute the CDF of the inter-arrival time between subsequent bearer instantiation requests at the MME, i.e., $F_\beta(\tau)$. We account for the fact that not all transitions to state $A$ by an MIoT source lead to a packet transmission: after a transition in state $A$ by an MIoT source, the probability of transmitting at least a packet is equal to $1\mathord{-}e^{-1}$. 
Thus, we  compute $F_\beta(\tau)$ considering that two subsequent transmissions in the system are separated by $z-1$ transitions to state $A$ without any transmission. Given the fact that the time for a transition to state $A$  is well described by an exponential distribution ($F_\alpha(\tau)$), we compute $F_\beta(\tau)$ as a sequence of $z$ i.i.d. exponentially distributed time intervals, i.e., an Erlang$(z,\lambda_\alpha)$ distribution, weighted by the probability that two subsequent  transmissions in the system are separated exactly by $z$ transitions to state $A$. Denoting the Erlang$(z, \lambda_\alpha)$ distribution with $E(z,\lambda_\alpha)$, we write: 
\begin{equation}
\label{eq_erlang}
F_\beta(\tau)   = \sum_{z=1}^{\infty} F_{E(z,\lambda_\alpha)}(\tau)(1-e^{-1})(e^{-1})^{z-1}.
\end{equation}

In (\ref{eq_erlang}), we remark once again that the probability of transmitting at least one packet in state $A$ is an input data to the model, and the specific value  in \cite{{laner2013traffic}} can be substituted with any arbitrary value  (even $1$, assuming that packets are always sent successfully upon visiting state $A$). Using the above results, we can prove the theorem below.
\begin{theorem}
\label{th_exponential_conditioned}
When the  IoT group offsets are independent of each other and  $Q$ grows large, $F_\beta(\tau)$ is given by:  
\begin{equation}
F_\beta(\tau)\mathord{=}1\mathord{-} e^{-\lambda_\beta \tau}
\end{equation}
with rate parameter $\lambda_\beta=\frac{Q(1-e^{-1})}{T}$.
\end{theorem}
\begin{proof}
The proof of the theorem can be found  in Appendix B in the Supplemental Material.
\end{proof}

The above result states that the inter-arrival time between bearer establishment requests at the MME follows an exponential distribution, which implies that {\em the number of requests that the MME, hence the EPC,  receives in a time interval follows a Poisson distribution}. 
This is  a key result that allows us to characterize first  the control overhead due to bearer establishment and forwarding, and then  the delay performance of the EPC.   
Note that the above result holds also in more general scenarios where there are aggregators relaying the data packets generated by the MIoT sources  (and requesting for  bearer instantiations) to the MME.

%

\section{EPC Model and Analysis} 
\label{sec:epc-model}

In this section, 
we begin by showing the results of our experimental study, which highlight the following facts: (i)  the  bearer establishment takes a deterministic amount of processing, (ii) the variation in the delay of EPC entities other than the MME is  negligible, (iii)  a PS  well mimics the MME serving policy. To perform our validation, we run and profile the  components of a real-world EPC implementation called OpenAirInterface (OAI) \cite{OAI}, as described in \sref{subsec:OAI_EPC}. Then, based on the above key observations, we analytically  characterize the EPC control overhead and, using a M/D/1-PS model, we derive an expression for the delay experienced by the MIoT traffic  within the EPC.

\subsection{Understanding the EPC through the OpenAirInterface  implementation}
\label{subsec:OAI_EPC}


The OAI EPC is an implementation of the cellular core network where the MME and the HSS  are implemented as separate entities, while the S-GW and the P-GW as a single unit (called SPGW). To investigate the interaction between the EPC and the IoT sources, we connected the OAI EPC to a software simulator of the Radio Access Network (RAN), called Open Air Interface Simulator (OAISIM). Herein, UEs and eNBs communicate with the OAI EPC through an Ethernet cable, sending and receiving control messages as if a real RAN was in place. 

The use of OAISIM implicitly creates some limitations to our experimental results, the most important one the fact that OAISIM supports a maximum of 3 simulated UEs in our setting. Nevertheless, we use OAI EPC and OAISIM for our study because it is an open-source controlled environment where the behavior of the EPC and the UEs can be controlled at the millisecond time-scale. Also, importantly, OAI EPC is compliant with Release 10  functionalities, and off-the-shelf smartphones can connect to the OAI EPC. 
Finally, we mention that, even if OAI EPC implements the standard bearer establishment procedure, which includes a superset of the messages exchanged between the EPC entities during the CIoT bearer instantiation procedure, below we report the results considering only the messages included in the CIoT procedure, as depicted in \fref{diagram}.

The total number of CPU  operations for each EPC entity, obtained by profiling the OAI EPC with the Callgrind tool from the Valgrind  suite \cite{Valgrind}, is depicted in \fref{experiments_results}. Therein, the number of users attached to the EPC varies from 1 to 3, and each data point has been obtained using $20$ runs, resulting in a 95\% confidence interval of up to $\pm0.7$\% of the plotted percentile values. Note also that, in each run, every UE performs one bearer establishment.  

\begin{figure}
\centering
	\includegraphics[width=1\columnwidth]{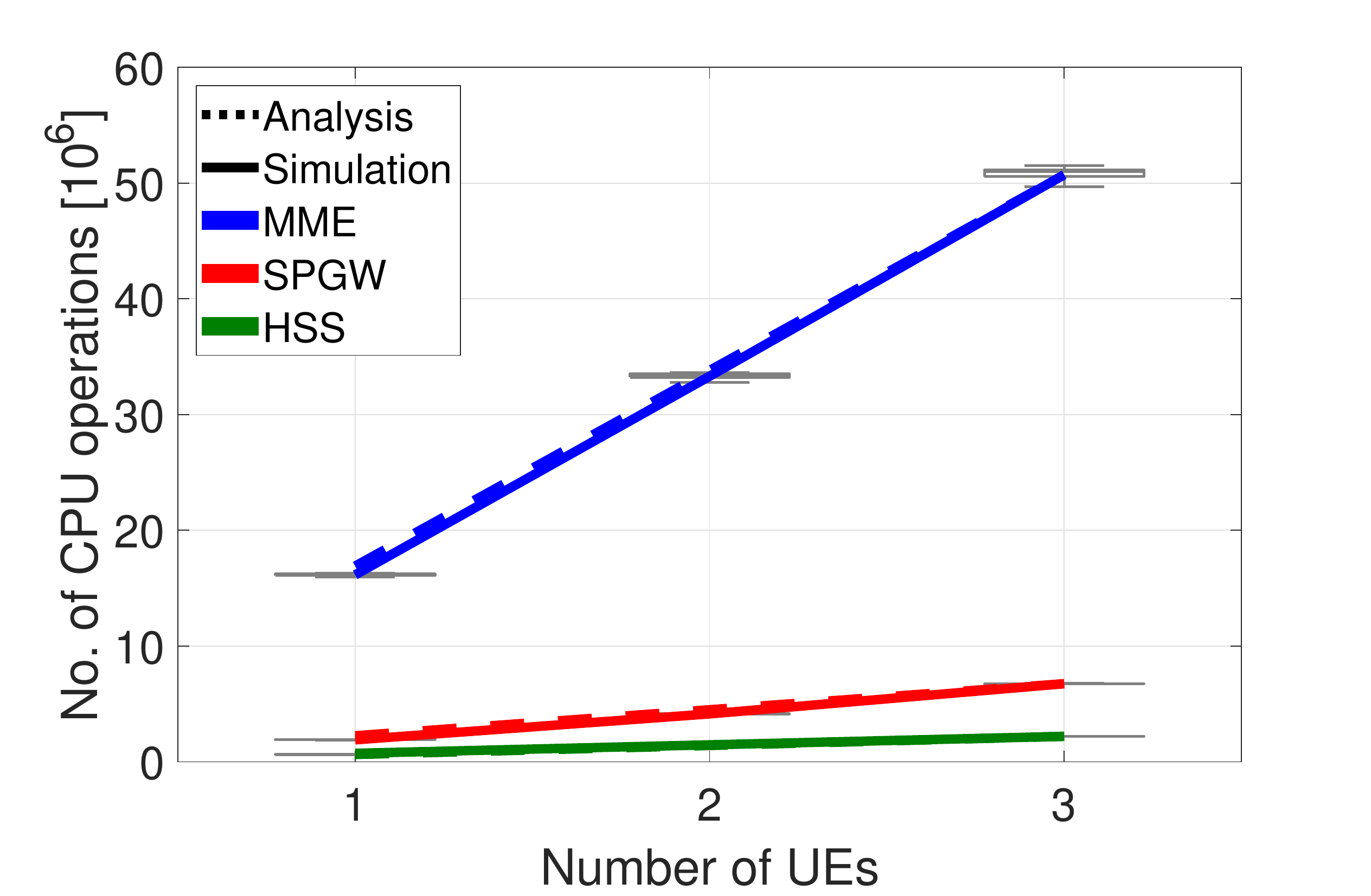}
\caption{Number of CPU  operations required by bearer establishment vs. number of UEs. }
\label{experiments_results}
\end{figure}

\fref{experiments_results} demonstrates that  the job size associated with a bearer establishment procedure is fixed and deterministic. This is shown by two facts: (1) the number of  operations required for a bearer establishment procedure grows linearly with the number of bearer instantiation requests, and (2)  the variation of the number of CPU operations required by the EPC entities across different runs is negligible. The former is further highlighted in the plot by the excellent match between the solid line, showing the experimental values, and 
the dotted line, which represents a linear fit whose slope is forced to the average number of CPU operations required by a single bearer instantiation.  
The latter fact, instead, can be  observed from the boxplots in \fref{experiments_results}, representing the 10-th and 90-th percentile of the CPU operations distribution: the variance is very small in all analyzed cases and for any of the EPC entities in the system.

The second important observation we can make by looking at \fref{experiments_results} is that the MME is the dominant component of the performance for the EPC: the number of operations required by any other entity is at most the 13\% of those needed by the MME.  
Given the fact that typical EPC implementations include entities with similar computational capability \cite{hasegawa2015joint}, as the  traffic load changes, it is fair to neglect the variations in the delay introduced by entities of the EPC other than the MME. 

We also performed dedicated experiments to grasp some insights on the policy used by the OAI EPC implementation to serve packets that are simultaneously queued at the different entities. In this set of experiments, in order to avoid interference, we isolate MME, SPGW, and HSS, assigning to each entity a dedicated CPU core of the PC acting as EPC. Note that all the UEs emulated through OAISIM make a bearer request almost simultaneously and it is not possible to determine beforehand their time of attachment. Therefore, users contend for the same resources during nearly the whole duration of the bearer establishment. 
\begin{table}[tb]
\caption{Bearer Establishment Time. 1 UE vs. 2 UEs \label{tab:bearer_establishment} }
\centering
\begin{tabular}{|| c | c | c ||} 
 \hline
 1 UE Bearer Time & 2 UEs Bearer Time &  2 UEs Bearer Time \\
 	 Average      & $1^{st}$ User - Avg.  &  $1^{st}$ User - PS Avg. \\
 \hline
 $0.84 \pm 0.02$ s & $1.63 \pm 0.03$ s & $0.87 \pm 0.02$ s \\ 
 \hline
\end{tabular}
\end{table}
The first and second column of Table~\ref{tab:bearer_establishment} show the average  (over 20 runs) and the standard deviation of the time elapsing between the first and the last packet processed by the MME with one UE and two UEs, respectively. In the case of two UEs, we only consider the data relative to the bearer establishment of the first UE. In the third column, we demonstrate that it is fair to assume that a PS policy is in place. Indeed, considering the time in the second column, and halving the time in which the procedures of the two UEs overlap, we obtain a value that is very close to the one in the first column.

Finally, we argue that our experimental results, although obtained for 3 users, have general validity. Indeed, 
\begin{itemize}
\item the MME uses a PS policy to serve the incoming traffic and the number of served users does not have any impact on the service policy of the system. Note that such an observation is consistent with the fact the PS policy closely emulates the behaviour of a multi-threaded application running on a virtual machine  instantiated on commodity hardware;
\item the number of CPU operations required by a bearer establishment procedure using CIoT optimization is deterministic at any entity (as also described in, e.g., \cite{prados2017modeling}), and does not depend on the number of on-going procedures;
\item the MME is the computational bottleneck of the EPC, which
is also evident given the load and the capacity values assigned to the EPC entities \cite{hasegawa2015joint}.
\end{itemize} 
It is therefore fair to consider that the assumptions we make, based on our experimental findings, still hold as the number of IoTs grows.

\subsection{Control Overhead and EPC Delay Characterization}\label{sec:delay_iot_traffic}

As discussed above, the bearer establishment procedure in \fref{diagram} requires a deterministic number of CPU operations. Then, at every entity $X$ involved in the procedure, each bearer instantiation is characterized by a fixed number of CPU operations $O_{X}$, which is the sum of the CPU operations required by the messages in \fref{diagram}. It follows that the mean number of CPU operations per second that entity $X$ has to perform is given by:  
$\EE[\mu_X] = \lambda_\beta O_{X}$.

Next, we  derive the pdf, $f_d(\tau)$, of the  interval between a bearer request and its completion, i.e., the time passing from the first to the last message in \fref{diagram}.  To this end, we exploit the fact that the inter-arrival time of bearer requests at the MME follows an exponential distribution, as well as the observations set out below, which have been derived through the experimental measurements. 
\begin{itemize}[leftmargin=0.8cm]
\item[{\em (a)}] The MME is the main bottleneck of the  control plane. As shown experimentally in \sref{subsec:OAI_EPC}, the computational load requested to the MME for a single bearer implementation is roughly one order of magnitude larger than the computational load requested to any other entity. This implies that the CPU utilization of entities other than the MME is very low and variations of the control message processing times can be neglected, i.e., they can be considered as constant. 

\item[{\em (b)}]  As shown above, the MME serving policy can be modeled through a PS discipline. 

\item[{\em (c)}]  It is fair to assume that the duration of a specific bearer instantiation procedure is very short compared to the timescale at which the MME load varies. When the utilization of the MME is high, and making scaling decisions is critical, the number of competing messages at the MME is high as well. Thus, during the short time-scale of a bearer instantiation (in the order of milliseconds), the difference between the number of incoming and outgoing messages at the MME is negligible if compared to the number of queued messages. It follows that the fraction of capacity assigned, according to the PS policy, to an MIoT bearer request does not vary throughout a bearer instantiation procedure and the processing time of each message belonging to the same bearer instantiation is roughly the same (as in an M/D/1-PS queue). Thus, each bearer request can be considered as a single job, even if composed of multiple subsequent messages, with a computational load equal to $O_{MME}$. 
\end{itemize}

Given the above observations and the result in \thref{th_exponential_conditioned}, we model the MME as an M/D/1-PS queue, where the deterministic service time depends on the capability of the MME, while the rate of arrivals of the bearer instantiation requests  is equal to $\lambda_\beta$, as reported in \thref{th_exponential_conditioned}.
Then $f_d(\tau)$ can be written as,
\begin{equation}
\label{delay}
f_d(\tau) = f_v(\tau) + K,
\end{equation}
where:
\begin{itemize} 
\item $f_v(\tau)$ is the pdf of the time spent by a bearer instantiation at the MME, i.e., the sojourn time of a job in the M/D/1-PS queue;
\item $K$ is the constant delay due to entities other than the MME (see our  observation {\em (a)} above), which can be  computed as:
\begin{equation}
\label{eq:K}
K = \frac{O_{UE}}{C_{UE}}+\frac{O_{eNB}}{C_{eNB}}+\frac{O_{HSS}}{C_{HSS}}+\frac{O_{S-GW}}{C_{S-GW}}+\frac{O_{P-GW}}{C_{P-GW}},
\end{equation}
 where $O_{X}$ is the total number of CPU operations that entity $X$ has to perform for each  bearer establishment, while $C_{X}$ is the computational capability of entity $X$, expressed in CPU operations per second. 
\end{itemize} 

To derive $f_v(t)$, we leverage  the results in \cite{egorova2006sojourn}, which, owing to the complexity of  computing such a distribution, provides the following approximation for the CDF: 
\begin{equation}
F_v(\tau) \approx \psi e^{-\gamma \tau} \,.
\label{MME_delay}
\end{equation}
In the above equation, $\psi$ is given by \cite{egorova2006sojourn}:
\begin{equation}
\psi = \frac{(1-\rho)(\lambda_\beta-\gamma)}{2\lambda_\beta (1-\rho)-\gamma\rho(2-\rho)}\,,
\end{equation}
where $\rho=\lambda_\beta D$ is the control traffic load at the MME, with $D=\frac{O_{MME}}{C_{MME}}$ being the deterministic service time of the bearer instantiation at the MME,  
and $\gamma$ is the only positive solution of  \cite[Eq.\,(3.2)]{egorova2006sojourn}. 

We remark that, given the pdf of the  time interval between a bearer instantiation request and its completion (i.e., $f_d(\tau)$), we can compute the pdf of the  delay that the control plane introduces in handling data packet forwarding at the MME when the CIoT optimization is supported. The derivation of the latter pdf implies considering only the messages in \fref{diagram} that are exchanged till the data packet transmission is completed. Then, based on our earlier observation {\em (c)} and given the number of CPU operations required by each message, 
we can obtain the pdf of the MIoT traffic latency by properly scaling $f_d(\tau)$. 

\begin{figure}
\centering
	\includegraphics[width=1\columnwidth]{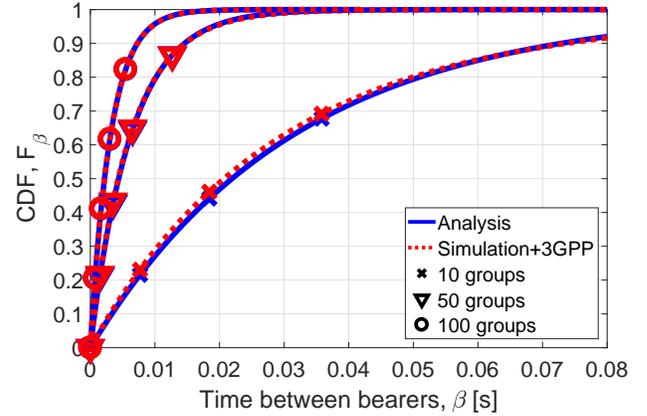}
\caption{Inter-arrival time distribution of bearer requests: analysis vs. simulation using the 3GPP traffic model.}
\label{arrival_validation}
\end{figure}

\begin{figure}
\centering
	\includegraphics[width=1\columnwidth]{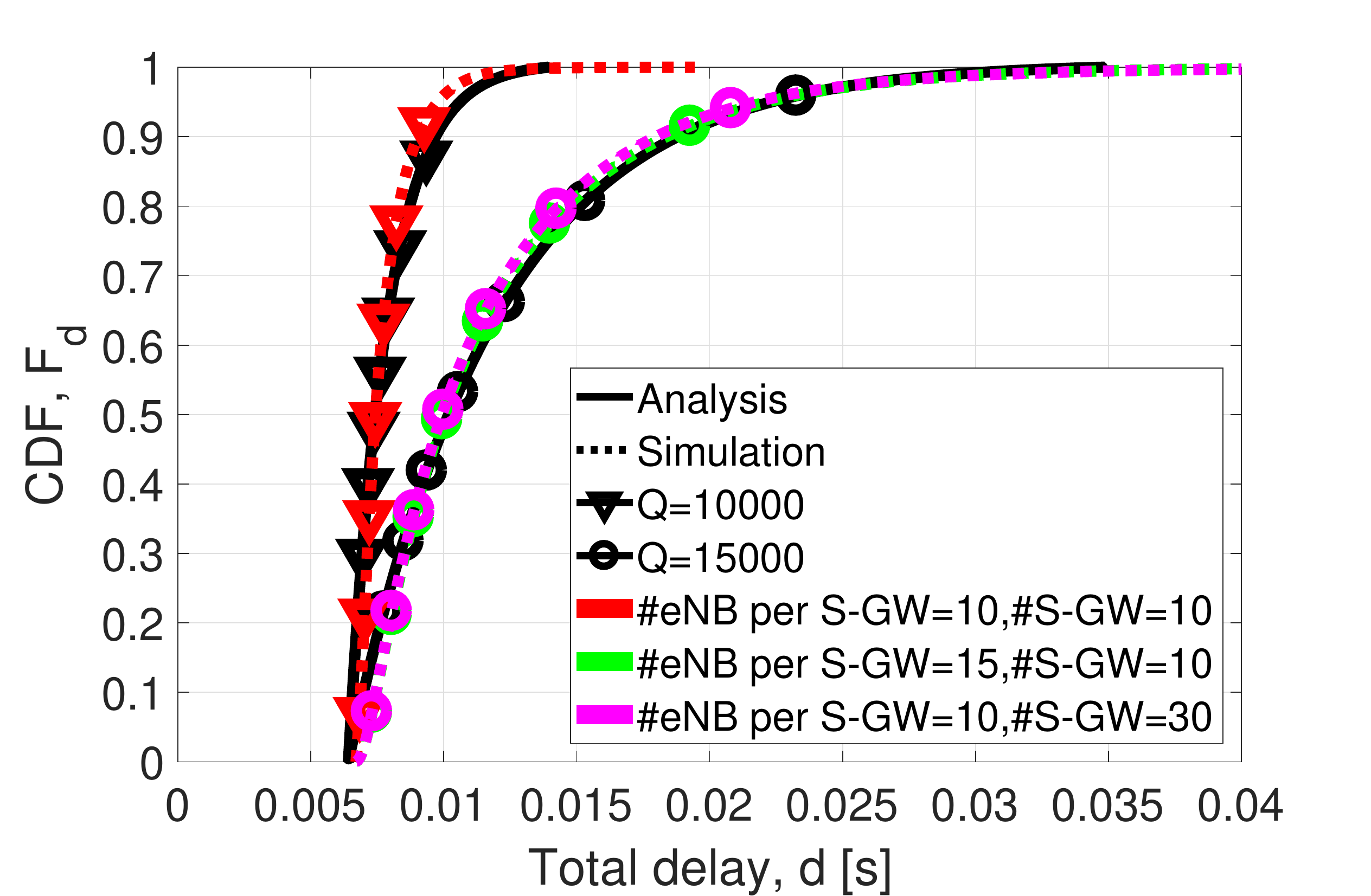}
\caption{Delay distribution: analytical  vs. simulation results, using the 3GPP traffic model.}
\label{delay_validation}
\end{figure}

\section{Model Validation and Exploitation}
\label{sec:results}

In the following, we show how the behaviors -- inter-arrival times and bearer instantiation delays -- predicted by our analysis match those yielded by extensive simulations, using both synthetic traffic models \cite{3gppepc}  (\sref{subsubsec:valexp_synth}) and real-world mobility traces (\sref{subsubsec:valexp_trace}). Furthermore, we demonstrate how our model can be leveraged in the dimensioning and management of vEPC networks handling MIoT traffic.


\subsection{3GPP synthetic traffic}
\label{subsubsec:valexp_synth}
We  developed a Matlab simulator that accurately implements the 3GPP traffic model described in \sref{subsec:iot_traffic_model}.
The parameters we used are as follows: $T=10$\,s (as specified in \cite{3gppepc}), $\delta = 10$\,$\mu s$, and group size equal to $50$ MIoT sources. 

{\bf Model validation.} Here we  validate the approximations introduced in our analysis as well as our main result in \sref{sec:input_MME} (i.e., the  inter-arrival time of bearer requests is exponentially distributed). To compare $F_\beta(\tau)$, computed as in \thref{th_exponential_conditioned}, to the CDF of the inter-arrival time of bearer requests at the MME in simulation, we performed  extensive  experiments, varying the number of groups in the scenario and the offsets between them $\omega_q$. In \fref{arrival_validation}, we present the results obtained with a specific set of offsets, as the number of groups served by the MME varies; however,  similar results have been also obtained changing the $\omega_q$~values.


With as few as $10$ groups served by the MME, \fref{arrival_validation} highlights that simulation and analytical results closely match, thus showing that the exponential  $F_\beta(\tau)$ captures very well the behavior of the 3GPP  traffic model presented in \sref{subsec:iot_traffic_model}. Furthermore, as expected, the match between the two curves improves as  the number of groups served by the MME grows. 

We now validate our delay model presented in \sref{sec:delay_iot_traffic}. We first remark that, for the analytical derivation of $f_d(\tau)$, we neglected the load due to the integrity check and decryption, at the MME. Indeed, while a single control message requires (roughly) one million floating-point operations \cite{prados2017modeling}, studies on commodity processors show that nowadays a $50$-byte packet (as in the case of IoT applications) requires few hundreds of floating-point operations for encryption/decryption \cite{EncryptionPerformance}.  In our simulations, instead, we account for data encryption/decryption as well as integrity check at the MME. 
Second, to compute the constant delay component of the delay distribution, $K$, in (\ref{eq:K}), we proceed as follows:  
\begin{itemize}
\item we obtained the number of CPU  operations, $O_X$, required at the  EPC entities by a bearer establishment through our experimental measurements described in \sref{subsec:OAI_EPC}, and 
\item we leveraged the work in \cite{hasegawa2015joint}, which provides the computational capability of the EPC entities, $C_X$, based on real-world data from a large mobile network operator.
\end{itemize}

Finally, in order to validate the analytical expression of $f_d(\tau)$, we extended our Matlab simulator to perform the whole procedure in \fref{diagram}, starting from the S1-AP Initial Message sent by the eNB. In our setup, all MIoT sources belonging to the same group, each containing $50$ MIoT sources, are attached to the same eNB. Several eNBs may be attached to the same S-GW, while all S-GWs are attached to the same P-GW. 
Except for the  RRC connection closing message sent by the eNB to the UE, all messages belonging to the same bearer instantiation travel sequentially between the involved entities, as foreseen by the CIoT optimization.   
Each entity is implemented as a PS server whose service rate matches the processing  capability provided in \cite{hasegawa2015joint}. 


\fref{delay_validation} shows the analytical and experimental $F_d(\tau)$ in different scenarios. Specifically, we present the results of the CIoT optimization for two different values of traffic load,  i.e., with $Q=10,000$ and $Q=15,000$. In the latter case, we also study two different configurations of the EPC to check whether changing the number of eNBs/S-GWs in the system has an impact on $F_d(\tau)$ or not. 

First, we observe that the CDF of the bearer instantiation delay computed  through (\ref{delay})-(\ref{MME_delay}) closely matches the experimental delay obtained via simulation -- a fact that is especially evident looking at the tail of the CDFs. This result proves that considering the whole bearer establishment handshake as a single job at the MME, plus a constant delay due to the other entities, is a valid approximation. Small differences between the analytical and experimental CDFs for low values of delay, are mainly due to the model in \cite{egorova2006sojourn}, used to approximate the sojourn time in an M/D/1-PS queue. Indeed, due to the complexity of the M/D/1-PS characterization, \cite{egorova2006sojourn} explicitly aims at modeling with higher accuracy the tail of the sojourn time CDF,  which is what most matters in delay sensitive applications. Second, for $Q=15,000$ the simulation results highlight that the two configurations with a different number of S-GWs provide exactly the same delay CDF, which validates our finding:  $F_d(\tau)$  depends only on the number of MIoT sources in the scenario, and it is not affected by variations in the number of eNBs and S-GWs. This confirms that the MME delay contribution dominates that of the other  EPC entities. 

{\bf Model exploitation.} 
We now show how our model can be used to develop efficient scaling algorithms for EPC networks serving MIoT traffic.
Let us consider the following case, reflecting, e.g., a smart factory or cloud robotics application \cite{C-robostics}, where  the delay introduced by the EPC should be less than $0.1$\,s with 0.99 probability. Since the delay performance depends on the number of MIoT sources served by the EPC and on the capability of the EPC entities, we need an algorithm  that, given the IoT traffic, scales the capability of the EPC entities according to the number of MIoT sources in the system.  Such an algorithm can leverage the analytical expression of $F_d(\tau)$ we obtained. 

As an example,  
we considered  a simple threshold-based algorithm, which, as the number of active IoT sources grows,  increases the computation capability of the EPC entities by $100\%$, and then by $150\%$, with respect to the initial value, depending on the MME delay predicted by our model (note that increasing the EPC capability by 100\%  can be realized by creating a new instance of  its components). As shown in \fref{delay_exploitation}, such an algorithm meets the target performance. The figure also reports the delay corresponding to the cases when the capability values $C_X$ are fixed to the initial value provided in \cite{hasegawa2015joint}, and to such a value increased by 100\%  or  by 150\%. 
Although more advanced scaling algorithms may be designed, we remark that, thanks to our model, even a simple threshold-based algorithm is able to meet the target delays and that our analysis, coupled with off-the-shelf virtualization tools like OpenStack, can be a key enabler to the support of IoT applications with  delay guarantees.
\begin{figure}
\centering
    \includegraphics[width=1\columnwidth]{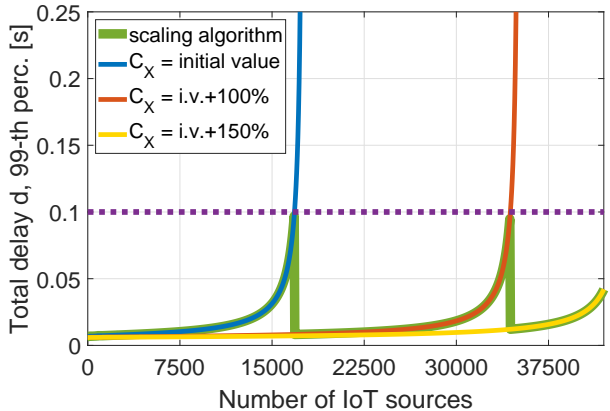}
\caption{Analysis exploitation:  $99-$th percentile of the bearer instantiation delay vs. number of MIoT sources. Performance obtained with:  the initial computational capability ($C_X$) of the EPC entities  set to the value provided in  \cite{hasegawa2015joint} (blue line), capability increased by 100\% (red line),  capability increased by 150\% (yellow line), and capability determined through the scaling algorithm (green line). Dashed, purple line: target value of the $99-$th percentile.}
\label{delay_exploitation}
\end{figure}

\subsection{Real-world trace}
\label{subsubsec:valexp_trace}
We now consider a real-world setting and leverage a large-scale mobility trace generated accounting for the MoST scenario~\cite{most}. The MoST scenario is a highly detailed representation of the mobility in the Monte Carlo urban area, including: (i)  a multi-layered road topology, with tunnel and bridges; (ii) multi-modal mobility, e.g., users driving to a parking lot and riding public transportation thence; (iii)  multiple types of coexisting  users, e.g., commuters and tourists. The scenario models the mobility of a total of $27,967$ users throughout an 8-hour period from 5~AM to 1~PM,
and includes a total of $607$ tagged points of interest (POIs) such as offices, restaurants, and tourist attractions. We assume that  every time a user visits or stops at one of the POIs,  a sensor, e.g., an identity-recognizing device,  is triggered, resulting in a packet transmission, hence,  a bearer instantiation request towards the MME serving the  area.

{\bf Model validation.}  Our first objective is to establish whether the inter-arrival time between bearer requests obtained experimentally matches the exponential distribution $F_\beta(\tau)$  obtained  through our analysis. 
To this end, we divided the time into 8 periods of one hour each, and computed the empirical distribution of the inter-arrival time of bearer instantiation requests in the MoST trace in every time period.  
The average arrival rates of bearer requests in the the various periods are very different, reflecting the daily fluctuations in  mobility.  Nevertheless, as exemplified in \fref{fig-KS}, the match between the analytical and the empirical distribution is excellent for all the time periods, proving that the  inter-arrival time of the bearer requests obtained from the MoST trace follows an exponential distribution as well. This confirms that our analysis holds also for applications that do not follow explicitly  the 3GPP traffic model described in \sref{subsec:iot_traffic_model}. 

\begin{figure}
\centering
	\includegraphics[width=1\columnwidth]{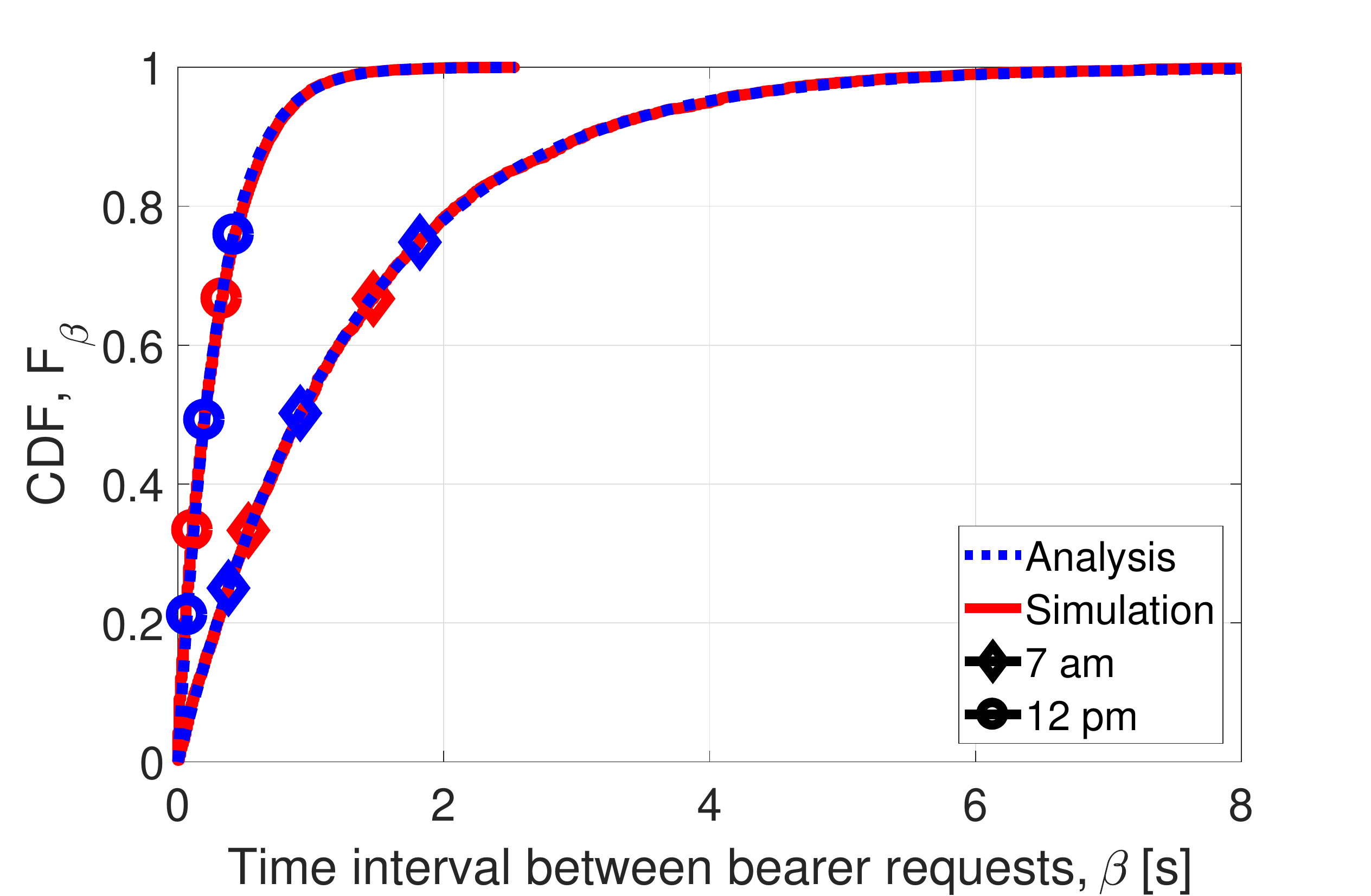}
\caption{MIoT trace: Comparison between analytical and  empirical distribution of  the inter-arrival time of bearer  requests at the MME, for two representative time periods.}
\label{fig-KS}
\end{figure}


Then we used the bearer requests  obtained from the MoST trace to evaluate if the delay distribution of the bearer request procedures can be approximated with $f_d(\tau)$ (as in \eref{delay}) also in this realistic IoT scenario. We fed to the previously mentioned Matlab simulator  the time instants of the bearer requests by sensors in the MoST trace and, since the number of bearer requests is rather small even in the rush hour, we reduced the EPC entities capability of one order of magnitude. The analytical and simulation results for the rush hour are compared in \fref{delay_IoT_app}, where the bearer establishment procedure delay is normalized by  $K$ (with $K$ given in \eref{eq:K}). 
Given the fact that the largest difference between the two CDFs (which happens for  low values of delay) is very small, we can conclude that our analysis well approximates  the behavior of the EPC when serving MIoT sources, also in the case of a realistic scenario  as the one of the MoST trace. 

{\bf Model exploitation.} We now present how our analytical results can be exploited under the MoST scenario. \fref{delay_exploitation2} shows the time evolution of the $99-$th percentile of the bearer instantiation delay, for different values of capability $C_X$. Considering a target performance of $0.1$\,s, we observe that the simple scaling algorithm employed when deriving \fref{delay_exploitation} (and which exploits our analytical results) successfully meets the delay requirement even under a sudden and very significant surge in the bearer request rate (see the black line in \fref{delay_exploitation2}).  

\begin{figure}
\centering
	\includegraphics[width=1\columnwidth]{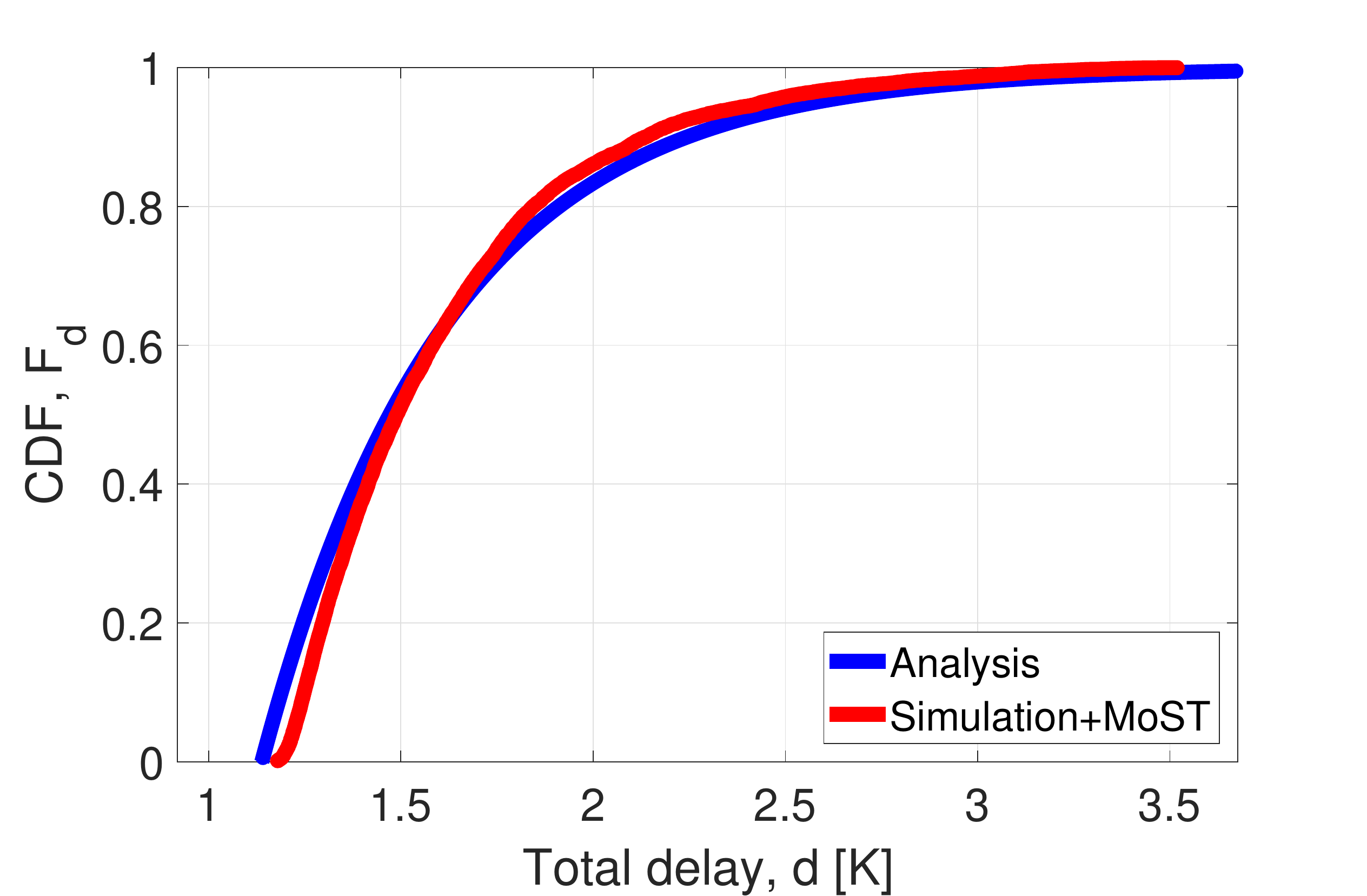}
\caption{Delay distribution of sensor traffic: Analysis vs. simulation, using the MoST trace.}
\label{delay_IoT_app}
\end{figure}

\begin{figure}
\centering
	\includegraphics[width=1\columnwidth]{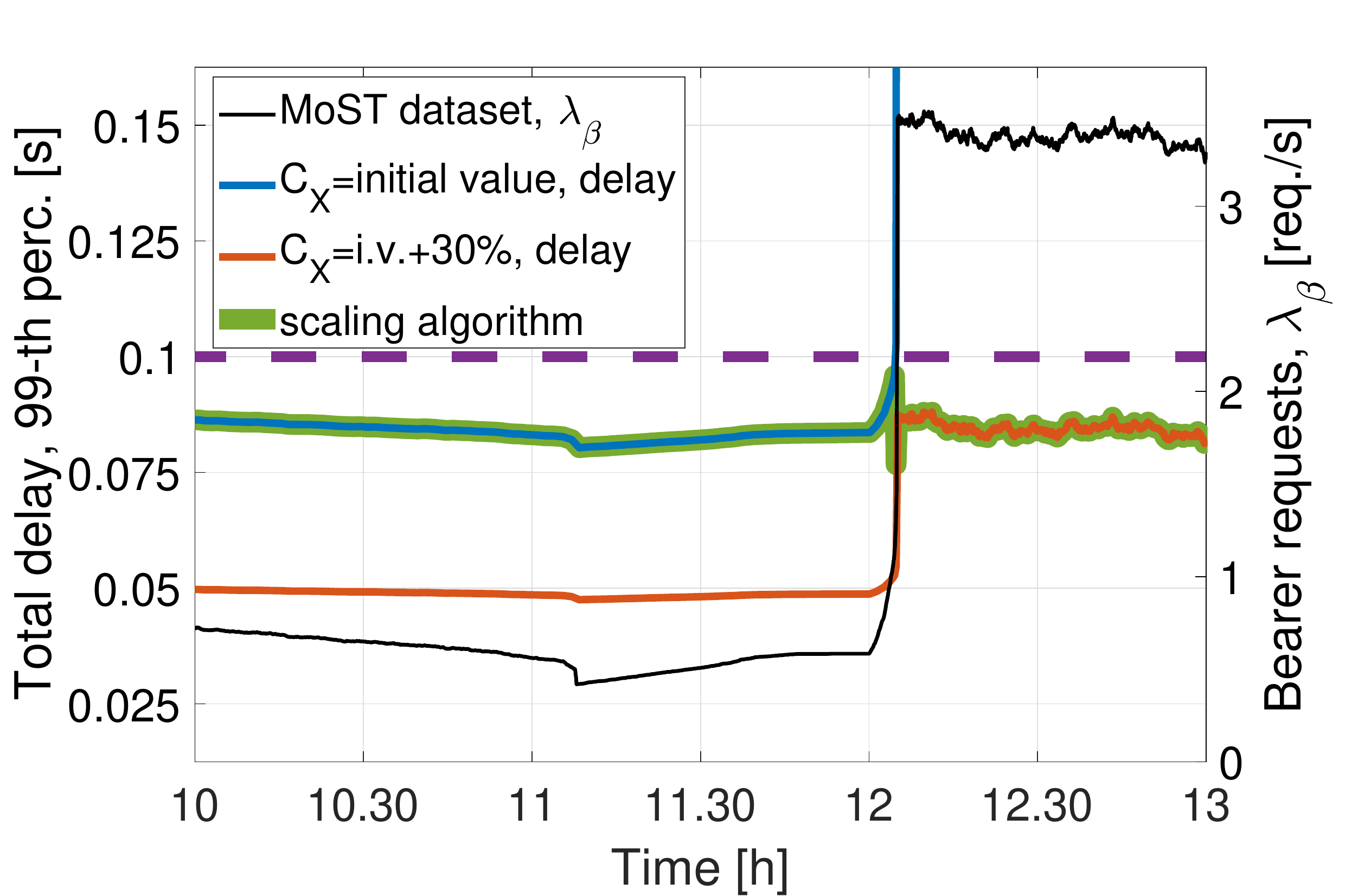}
\caption{Analysis exploitation: $99-$th percentile of the bearer instantiation delay vs. time, when the MoST trace is used. Black line: bearer request arrival rate from the trace; dashed, purple line: target value of the $99-$th percentile. Performance obtained with:   the initial computational capability of the EPC entities (blue line), capability increased by 30\% (red line),   capability determined through the scaling algorithm (green line).}
\label{delay_exploitation2}
\end{figure}

\section{Related Work}\label{sec:related} 
IoT support through cellular networks has recently attracted
significant attention by both the scientific community and the
standardization fora.

A first body of works deal with the requirements posed by IoT scenarios, and how 5G networks can cope with them. As an example, \cite{latencyiot} quantifies the latency and capability requirements for the main IoT scenarios, from factory automation to parking machines, and discusses the improvement needed to both the radio access and core networks. The CIoT optimization and the role of MME are, however, not accounted for in~\cite{latencyiot}, which mainly focuses on the SGW/PGW gateways. \cite{C-robostics}~has a narrower focus, namely, cloud robotics, and presents a working prototype; however, the latency introduced by the core networks and the entities therein is not taken into account. 
\cite{beyene2017nb} and \cite{darsena2019cloud}, instead, present the impact on the radio access, e.g., on Cloud-RAN, of IoT-specific physical layers, such as NB-IoT. 
Note that these works tackle the  IoT support in cellular networks from the perspective of centralized IoT transmission coding/decoding, but they do not take into account IoT optimized control procedures, such as the CIoT.

Among the studies that do account for the core network, most, including~\cite{orchestration,iotmulticast}, envision a virtualized  network, where  network functions are implemented through VNFs. Unlike our work, \cite{orchestration} does not specifically target EPC or any of its entities. The authors of~\cite{iotmulticast}, focusing on multicast traffic in IoT scenarios, specifically study the MME delay. Their proposed solution is to endow the SGW with some of the MME tasks, the opposite of the CIoT optimization we consider in our study. 

The use of NFV and SDN, for the implementation of the
EPC under massive IoT
traffic conditions, has been discussed in \cite{jain2016comparison},
while 
enhancements to  the standard EPC can be 
found, e.g.,  in \cite{nagendra2016lte}. That work  introduces  new entities
in the network architecture, which are specifically 
devoted to the IoT support. Importantly, although such solutions yield
a remarkable performance improvement, they inevitably involve  significant changes to
the standard.

Analytical models of IoT systems have been developed for specific application use cases, like management \cite{7336570}, opportunistic crowd sensing in vehicular scenarios \cite{7857676}, or ambient backscatter devices \cite{7820135}.
Other works have presented theoretical models for the study of networking aspects such as  the performance of middleware protocols \cite{7997451}, implemented between the application and the transport layer, or of the random access procedure in NB-IoT \cite{8288195,8605340}.
Note, however, that none of the above works investigates the critical role of the EPC control plane in IoT-based systems; 
indeed few studies exist on the characterization of the overhead and service delay when the EPC handles massive IoT
data traffic. 
In this context, the studies that are the most relevant to ours are \cite{hasegawa2015joint} and its extension \cite{abe2017design}, which present a scheme for aggregating multiple IoT bearers and analyze the gain that is obtained with respect to the standard procedure. Such works, however, are based on deterministic inter-packet transmission time for MIoT sources and do not address the most recent and efficient 3GPP specifications for IoT support. 
Likewise, the study in \cite{prados2017modeling} analytically
evaluates the EPC control procedures (not specific to CIoT optimization) considering a simple IoT traffic model, coexisting with other cellular traffic sources. Furthermore, unlike our work, both \cite{prados2017modeling} and \cite{abe2017design} derive only  the {\em average} processing latency of standard bearer establishments. A more comprehensive study on EPC control procedures has been presented in our conference paper \cite{nostro-globecom}, which, however, does not address any delay analysis. 

To the best of our knowledge, our work is the first one presenting the delay distribution under the CIoT optimization and deriving an exponential inter-arrival time for MIoT bearer establishments at the MME. It is important to stress that, in our paper, exponential inter-arrival times are not an assumption, but the result of the analysis in \sref{sec:input_MME}, which is based on the 3GPP MIoT traffic model \cite{3gppmtc}  and has been validated in Sec.\,\ref{sec:results}. A sketch of our work has been presented in our poster publication \cite{nostro-poster}. Finally, we remark that the goal of our work differs significantly from that of  \cite{laner2013traffic}, which presents the Markov Modulated Poisson Process (MMPP) model for individual IoT sources that  we adopt to develop our analysis and that is in accordance with the 3GPP traffic model for IoT. Indeed, \cite{laner2013traffic} investigates large-scale IoT scenarios via simulation only: it does not present any analytical model of the EPC or of its control procedures under IoT traffic support.

\section{Conclusions}\label{sec:conclusions} 

Effectively dimensioning the cellular core network and evaluating its latency performance are crucial tasks for the support of massive IoT applications. 
Observing that the MME has a major impact on the control-plane latency, we characterized the statistics of  the latency it introduces. In particular, we derived closed-form expressions that link the number of IoT devices to the inter-arrival times of bearer instantiation requests. Then, leveraging these results and an M/D/1-PS queue model of the MME, we characterized the latency experienced by IoT traffic in the cellular core. Importantly, our analysis is also based on findings obtained by measuring and profiling the performance of a real-world EPC implementation. 
Furthermore, using both the 3GPP traffic model and a real-world, large-scale mobility trace, we validated via  extensive simulations  the distribution of request inter-arrival times and of the IoT traffic latency.  

We demonstrated that our model and results can be exploited to dimension the computation capabilities of the  entities of the virtualized cellular core as the offered IoT traffic load varies.
Future work will leverage  our results to design more advanced algorithms for the scaling of the resources allocated to  virtualized EPC entities as the traffic load varies, and assess their performance using real-world EPC implementations.

\section*{Acknowledgments}
This work was supported by the European Commission through the H2020 project 5G-EVE (Project ID 815074). The work of Christian Vitale was also supported by the European Union’s Horizon 2020 Research and Innovation Programme under Grant 739551 (KIOS CoE) and from the Republic of Cyprus through the Directorate General for European Programmes, Coordination, and Development.

\bibliographystyle{IEEEtran}
\bibliography{vepc}

\end{document}